\documentclass[parindent=0pt]{article}
\usepackage{booktabs}
\usepackage{makecell}
\usepackage{subcaption}

\usepackage[utf8]{inputenc}     % for éô
\usepackage[english]{babel}     % for proper word breaking at line ends
\usepackage[a4paper, left=1in, right=1in, top=1.25in, bottom=1.25in]{geometry}
\usepackage{graphicx}           % for ?

\usepackage{amsmath,amssymb}    % for better equations
\usepackage{amsthm}             % for better theorem styles
\usepackage{mathtools}          % for greek math symbol formatting
\usepackage{enumitem}           % for control of 'enumerate' numbering
\usepackage{listings}           % for control of 'itemize' spacing
\usepackage{todonotes}          % for clear TODO notes
\usepackage{hyperref}           % page numbers and '\ref's become clickable
\usepackage{float}
\usepackage{amsthm}
\usepackage{amsmath}
\usepackage{amscd}
\usepackage{amssymb}
\usepackage{mathrsfs}
\usepackage{pdfpages}
\usepackage{setspace}
\usepackage{physics}
\usepackage{tcolorbox}

\usepackage{xcolor}
\usepackage{tikz}
\usepackage[toc,page,title,titletoc,header]{appendix}
 \usepackage{url}
 \usepackage{babel} 
\usepackage{csquotes}
\usepackage[f]{esvect}
\usepackage{blindtext}

\newcommand{\E}{\mathbb{E}}
\newcommand{\C}{\mathbb{C}}

\newcommand{\supp}{\mbox{supp}}

\newcommand{\g}{\mathfrak{g}}

\newcommand{\diag}{\mbox{diag}}

\newcommand{\oo}{\mathfrak{o}}

\let\phi=\varphi
\let\epsilon=\varepsilon

\newtheorem{theorem}{Theorem}

\newtheorem{proposition}[theorem]{Proposition}
\newtheorem{definition}[theorem]{Definition}
\newtheorem{conjecture}[theorem]{Conjecture}
\newtheorem{result}{Result}

\usepackage{authblk} % Pour gérer les affiliations multiples

\newcommand{\adots}{\mathinner{\mkern1mu \raise1pt
\hbox{.} \mkern2mu \raise4pt \hbox{.} \mkern2mu \raise7pt
\vbox{\kern7pt \hbox{.}} \mkern1mu}}
\def\[{[\![}
\def\]{]\!]}

\newlist{regimes}{enumerate}{1}
\setlist[regimes]{label={\textbullet\ Regime (\arabic*):}, leftmargin=*}

\title{%
   Spectral boundaries of deterministic matrices deformed by rotationally invariant random non-Hermitian ensembles
}

  \author[1,2]{Pierre Bousseyroux\thanks{Email: pierre.bousseyroux@polytechnique.edu}}
\author[3]{Marc Potters}

\affil[1]{Econophysics Lab, Institut Louis Bachelier, 28 Pl. de la Bourse, Palais Brongniart, 75002 Paris, France}
\affil[2]{LadHyX, UMR CNRS 7646, Ecole Polytechnique, Institut Polytechnique de Paris, 91128 Palaiseau, France}
\affil[3]{Capital Fund Management, Paris, France}

\newenvironment{remarks}{
  \par\vspace{1ex}
  \noindent\textbf{Remarks.}\begin{itemize}\setlength\itemsep{0pt}}
  {\end{itemize}\par\vspace{1ex}}

\makeatletter
\renewcommand{\@fnsymbol}[1]{%
  \ifcase#1\or
    \ensuremath{\dagger}\or
    \ensuremath{\ddagger}\or
    \ensuremath{\mathsection}\or
    \ensuremath{\mathparagraph}\or
    \ensuremath{\|}\or
    \ensuremath{\dagger\dagger}\or
    \ensuremath{\ddagger\ddagger}\or
    \ensuremath{\mathsection\mathsection}\or
    \ensuremath{\mathparagraph\mathparagraph}\else
    \@ctrerr
  \fi}
\makeatother

\begin{document}
\maketitle
\begin{abstract}
One of the great miracles of random matrix theory is that, in the $N \to \infty$ limit, many otherwise intractable matrix problems with horrendously complicated finite-$N$ expressions admit remarkably simple and elegant asymptotic solutions. In this paper, we illustrate this phenomenon in the context of \emph{spectral boundaries} (or spectral edges) for deformed random matrices. Specifically, we consider matrices of the form $\vb{A} + \vb{B}$, where $\vb{A}$ is a deterministic $N\times N$ matrix (not necessarily Hermitian) and $\vb{B}$ is a rotationally invariant random matrix. In the large-$N$ limit, we show that the complex eigenvalue distribution of $\vb{A} + \vb{B}$ satisfies remarkably simple boundary equations that depend on the $\mathcal{R}_1$ and $\mathcal{R}_2$ transforms of $\vb{B}$ defined in~\cite{bousseyroux1}. We illustrate our results on several explicit random matrix ensembles and support them with numerical simulations.
\end{abstract}

The limiting spectral distribution of $\vb{A}+\vb{B}$ is known only in a few special cases. When both $\vb{A}$ and $\vb{B}$ are Hermitian, the eigenvalues of $\vb{A}+\vb{B}$ are real, and their distribution can be obtained through Voiculescu’s additive free convolution~\cite{Voiculescu1986addition}. In this setting, the classical $R$-transform provides explicit formulas for the limiting eigenvalue density and its support. In contrast, for non-Hermitian deformations $\vb{A}+\vb{B}$, much less was known until relatively recently. The study of sums of non-Hermitian random matrices has a rich history, and several notions of $R$-transforms have been introduced to deal with such problems~\cite{JanikEtAlPRE1997, janik1997non, FeinbergZeeNPB501, FeinbergZeeNPB504, ChalkerWangPRL1997, Rodgers2010}. In the present work, we rely instead on the framework summarized and developed in~\cite{bousseyroux1}.

From a mathematical standpoint, the appropriate object for non-normal spectra is the \emph{Brown measure}, introduced by L.~G.~Brown and extended—also to certain unbounded operators—by Haagerup and Schultz~\cite{brown1986lidskii,haagerup2007brown}. Explicit descriptions of the Brown measure (or at least of its support and spectral edges) for deformations $\vb{A}+\vb{B}$ are available only in a few families, organized chiefly by the law of $\vb{B}$: (i) when $\vb{B}$ is Ginibre (i.i.d.\ complex entries with variance $1/N$), several works analyze the Brown measure and outliers of $\vb{A}+\vb{B}$ beyond the circular-law baseline~\cite{JanikEtAlPRE1997,janik1997non,FeinbergZeeNPB501,FeinbergZeeNPB504,ChalkerWangPRL1997,Rodgers2010, sarapin2025limiting, liu2022phase, fyodorov2025kac}; (ii) when $\vb{B}$ is elliptic in the sense of Girko \cite{girko1986elliptic}, deformations $\vb{A}+\vb{B}$ admit tractable descriptions of their limiting support and edges~\cite{Zhong2021EllipticPlusFree,Ho2022SelfAdjointPlusElliptic,HallHo2022ImagSemicirc}; and (iii) when $\vb{B}$ is bi-invariant \footnote{We say that a random matrix $\vb{M}$ is bi-invariant if $\vb{M}$ and $\vb{U}\,\vb{M}\,\vb{V}$ have the same distribution for all unitary matrices $\vb{U}$ and $\vb{V}$.}, the Brown measure of $\vb{A}+\vb{B}$ is also well understood~\cite{haagerup2000brown,bercovici2022brown,HoYinZhong2025OutliersFullRankSingleRing,HoZhong2025DeformedSingleRing}. Despite notable examples—see also the illustrative, problem-specific constructions of Biane and Lehner~\cite{BianeLehner2001}—general closed-form formulas for the full Brown measure of $\vb{A}+\vb{B}$ remain scarce; this motivates our focus on spectral boundaries.

In all these settings, explicit formulas for the full Brown measure often become prohibitively complicated. In this paper, we focus instead on the \emph{spectral boundaries} themselves.
We assume that the random matrices considered in this work admit a continuous spectral distribution, denoted by $\rho$, in the complex plane. This distribution is defined as the limit of the empirical spectral distribution of the complex eigenvalues in the high-dimensional limit. We denote by $\supp(\rho)$ the support of $\rho$, that is, the set of points $z \in \C$ such that $\rho(z) \neq 0$, and by $\partial \supp(\rho)$ its boundary, which precisely designates the spectral edges.
Beyond their theoretical interest, spectral boundaries play a key role in stability analysis, as the location of spectral edges often signals transitions such as the onset of chaos or loss of stability. Our goal is to develop a general theory for computing the support and its boundaries for deformations $\vb{A}+\vb{B}$ in the non-Hermitian setting, assuming that $\vb{B}$ is a rotationally invariant matrix\footnote{In the sense that $\vb{M}$ and $\vb{U}\vb{M}\vb{U}^*$ have the same distribution for every unitary matrix $\vb{U}$.}.
We provide a unifying framework to describe and compute the spectral edges of $\vb{A}+\vb{B}$ in full generality. Note that the case where $\vb{B}$ is a finite-rank normal operator has been studied in \cite{bousseyroux2}\footnote{An operator $\vb{B}$ is said to be \emph{normal} if it commutes with its adjoint, i.e., $\vb{B}\vb{B}^* = \vb{B}^*\vb{B}$, where $\vb{B}^* := \overline{\vb{B}}^{T}$ denotes the conjugate transpose of $\vb{B}$.}.

We will make use of the notions and notations introduced in~\cite{bousseyroux1}, to which we refer the reader, since the proofs presented here rely extensively on that reference. We introduce the following notations. Let $\vb{M}$ denote a large random matrix. We define
\begin{equation}\label{defh}
h_{\vb{M}}(z) = \tau\!\left([(z\vb{1} - \vb{M})(z\vb{1} - \vb{M})^*]^{-1}\right),
\end{equation}
and
\begin{equation}
\g_{\vb{M}}(z) = \tau\!\left((z\vb{1} - \vb{M})^{-1}\right),
\end{equation}
where $\vb{1}$ denotes the identity matrix and $\tau = \lim_{N \to +\infty} \frac{\tr}{N}$ is the normalized trace. We will also make use of the $\mathcal{R}_1$- and $\mathcal{R}_2$-transforms of the matrix $\vb{B}$, as well as their associated multivalued functions. Their definitions, given in \cite{bousseyroux1}, are recalled in Section~\ref{model_matrices}, where we compute them for several classes of matrices.

\section{Main results}

We can now state the main result which is proved in Appendix \ref{first_result}.

\begin{result}\label{main_result}
Let $\vb{A}$ be a large deterministic matrix and $\vb{B}$ a rotationally invariant random matrix.  
Then, the spectral boundaries of $\vb{A} + \vb{B}$ can be written as follows:

\begin{itemize}
    \item \textbf{boundary of type 1:} As the image of the points $x \in \C$ satisfying
    \begin{equation}\label{sol1}
        \partial_{\alpha} \mathcal{R}_{1, \vb{B}}(0, \g_{\vb{A}}(x)) = \frac{1}{h_{\vb{A}}(x)},
    \end{equation}
    where $\mathcal{R}_1 \in \widetilde{\mathcal{R}}_{1, {\vb{B}}}$, under the map
    \begin{equation}
        x \mapsto x + \mathcal{R}_{2, \vb{B}}(0, \g_{\vb{A}}(x)).
    \end{equation}

    \item \textbf{boundary of type 2:} By the points $z \in \C$ satisfying
    \begin{equation}\label{edge2}
        \tau(\vb{A}\vb{A}^*) - |\tau(\vb{A})|^2 = \frac{1}{h_{\vb{B}}(z - \tau(\vb{A}))}.
    \end{equation}
\end{itemize}

\end{result}

\begin{remarks}
\item There are two types of spectral edges, and we note that the second one, described by equation~\eqref{edge2}, provides a rather universal spectral edge, which depends on the matrix $\vb{A}$ only through the two quantities $\tau(\vb{A})$ and $\tau(\vb{A}\vb{A}^*)$.

\item Be careful with the interpretation of the result: we are not claiming that the curves given by the result will necessarily be spectral edges, but rather that the spectral edges of $\vb{A} + \vb{B}$ will be described in this way.

\item In the case where $\vb{B}$ is a bi-invariant matrix, using the single-ring theorem~\cite{FeinbergZeeNPB501, guionnet2011single}, we know that the support of the limiting spectral distribution of $\vb{B}$ is a ring, whose inner and outer radii are denoted by $r_{-,\vb{B}}$ and $r_{+,\vb{B}}$, respectively. The possible spectral edges are then given by
\begin{equation}\label{bi-invariant-case}
    r_{+,\vb{B}}^2 = \frac{1}{h_{\vb{A}}(z)},
\end{equation}
and by a circle centered at $\tau(\vb{A})$ with radius $\sqrt{r_{-, \vb{B}}^2 - \tau(\vb{A}\vb{A}^*)}$, whenever this quantity is well defined. Note that these spectral edges are highly universal with respect to $\vb{B}$ and depend only on $r_{+,\vb{B}}$ and $r_{-,\vb{B}}$. Example~\ref{sub:lemniscate} in Section~\ref{examples} will illustrate this property.

\item When $\vb{B}$ is bi-invariant and $\vb{A}$ is Hermitian, the previous remark implies that the spectral boundary is given by the set of complex numbers $z \in \C$ satisfying
\begin{equation}\label{formula_bi}
    r_{+,\vb{B}}^2\, \Im\bigl(\g_{\vb{A}}(z)\bigr) = -\,\Im(z).
\end{equation}

Let us now examine the limit as $r_{+,\vb{B}} \to 0$.  
Denote by $\rho_{\vb{A}}$ the limiting spectral density of the real eigenvalues of $\vb{A}$, supported on the compact interval $[\lambda_{-}, \lambda_{+}]$.  
In this regime, one finds
\begin{equation}\label{traceMP}
    |\Im(z)| \underset{r_{+,\vb{B}}\to 0}{\sim} r_{+,\vb{B}}^2\,\pi\,\rho_{\vb{A}}(\Re(z)).
\end{equation}
If $\Re(z) - \lambda_{+} \gg \Im(z)$, then
\begin{equation}
    \Re(z) - \lambda_{+} \underset{r_{+,\vb{B}}\to 0}{\sim} \frac{r_{+,\vb{B}}^4\,\pi^2\,C^2}{2},
\end{equation}
provided that
\begin{equation}
    \rho_{\vb{A}}(\lambda) \underset{\lambda \to \lambda_{+}}{\sim} C\,\sqrt{\lambda_{+} - \lambda},
    \qquad C>0.
\end{equation}

The denser the eigenvalues are near the edge, the more rapidly the spectrum is pushed away from the real axis.  
We observe that the imaginary part of the boundary departs from the real axis at order $r_{+,\vb{B}}^2$, whereas the real part shifts at order $r_{+,\vb{B}}^4$. We illustrate the content of this remark in Example~\ref{sub:airplane}, presented in Sec.~\ref{examples}.

\end{remarks}

Note that the previous result does not satisfactorily address the case $\vb{A} = 0$.  The next result analyzes this situation in detail and the proof is given in Appendix \ref{second_result}.

\begin{result}\label{thm:bordspectre}
Let $\vb{B}$ be a large random matrix a priori non-Hermitian. The spectral boundaries of $\vb{B}$ can be written as follows:

\begin{itemize}
    \item \textbf{boundary of type 1:} the image of the points $x$ satisfying
    \begin{equation}\label{boundary1}
        \partial_{\alpha}\mathcal{R}_{1, \vb{B}}(0, 1/x) = |x|^2,
    \end{equation}
    under the map
    \begin{equation}\label{map}
        x \longmapsto x + \mathcal{R}_{2, \vb{B}}(0, 1/x),
    \end{equation}
    where $\mathcal{R}_{1, \vb{B}}$ and $\mathcal{R}_{2, \vb{B}}$ are suitable determinations of $\widetilde{\mathcal{R}}_{1, \vb{B}}$ and $\widetilde{\mathcal{R}}_{2, \vb{B}}$.

    \item \textbf{boundary of type 2:} the points $z \in \mathbb{C}$ such that
    \begin{equation}\label{boundary2}
        \mathcal{R}_{1, \vb{B}}(\alpha, - (z - \tau(\vb{B})) \alpha^2)
        \underset{\alpha\to 0}{=}
        \frac{-1}{\alpha} - |z - \tau(\vb{B})|^2 \alpha + o(\alpha),
    \end{equation}
    for a certain determination of $\mathcal{R}_{1, \vb{B}}$.
\end{itemize}
\end{result}

\begin{remarks}
\item Once again, we are not claiming that the curves given by the result will necessarily be spectral edges of $\vb{B}$, but rather that the spectral edges will be described in this way. 
 \item In the case where $\vb{B}$ is bi-invariant, we deduce that there exists a first branch of $\mathcal{R}_{1, \vb{B}}$ such that
\begin{equation}
    \mathcal{R}_{1, \vb{B}}(\alpha) \underset{\alpha\to 0}{\sim} r_{+, \vb{B}}^2 \,\alpha,
\end{equation}
and a second one (in the case where the ring is non-degenerate) such that
\begin{equation}
    \mathcal{R}_{1, \vb{B}}(\alpha) \underset{\alpha\to 0}{=} -\frac{1}{\alpha} - r_{-, \vb{B}}^2 \,\alpha.
\end{equation}

\item As explained in~\cite{bousseyroux1}, $\widetilde{\mathcal{R}}_1$ is additive for rotationally invariant matrices. In particular, if we take two bi-invariant matrices $\vb{A}$ and $\vb{B}$, then applying the previous statement together with the preceding item yields the addition rules for the inner and outer radii of $\vb{A} + \vb{B}$.

Let $\vb{A}$ and $\vb{B}$ be bi-invariant matrices. Then
\begin{equation}
    r_{+,\,\vb{A} + \vb{B}}^2 = r_{+,\,\vb{A}}^2 + r_{+,\,\vb{B}}^2,
\end{equation}
and
\begin{equation}
    r_{-,\,\vb{A} + \vb{B}}^2 = \max\!\Bigl( r_{-,\,\vb{A}}^2 - r_{+,\,\vb{B}}^2,\; r_{-,\,\vb{B}}^2 - r_{+,\,\vb{A}}^2,\; 0 \Bigr).
\end{equation}

\item Using Result~4 of~\cite{bousseyroux1}, the boundary of type~\(2\) describes the boundary of a charge-free region where the field \(\g_{\vb{B}}(z)\) vanishes. We have chosen to express the theorem in terms of \(\mathcal{R}_{1,\vb{B}}\), but one may also express it in terms of \(\vb{B}^{-1}\). Studying a spectral boundary of \(\vb{B}\) for which the field \(\g_{\vb{B}}(z)\) vanishes in the charge-free region amounts to studying a spectral boundary of \(\vb{B}^{-1}\) for which the field \(\g_{\vb{B}^{-1}}(z)\) is equal, in a neighborhood of this boundary, to \(1/z\). Thus, by Result~4 of~\cite{bousseyroux1}, this corresponds to a region where \(\mathcal{R}_2 = 0\). One may therefore apply equation~\eqref{boundary1} to \(\vb{B}^{-1}\) to obtain that the boundaries of type~\(2\) of \(\vb{B}\), given by equation~\eqref{boundary2}, can equivalently be described by the points \(z\in\mathbb{C}\) satisfying
\begin{equation}
    \partial_{\alpha}\mathcal{R}_{1,\vb{B}^{-1}}(0,z)
    =
    \frac{1}{|z|^2},
\end{equation}
where \(\mathcal{R}_{1,\vb{B}^{-1}}\) is an appropriate branch.

\end{remarks}

Section~\ref{model_matrices} introduces several random matrix ensembles and derives the quantities needed to apply Results~\ref{main_result} and~\ref{thm:bordspectre}. Section~\ref{examples} combines these ensembles to construct explicit examples and provides numerical illustrations of our results.

\section{Model random matrix ensembles}\label{model_matrices}

We now introduce the basic random matrix ensembles that will be used throughout the paper to test and illustrate our results.

\begin{definition}[Complex Ginibre ensemble]\label{def:ginibre}
The matrix $\vb{G}$ denotes an $N\times N$ complex Ginibre matrix, i.e.\ a non-Hermitian random matrix with i.i.d.\ complex Gaussian entries of zero mean and variance $1/N$ \cite{girko1985circular}.
\end{definition}

\begin{definition}[Elliptic Ginibre ensemble]\label{def:elliptic_ginibre}
The elliptic Ginibre ensemble, denoted $\vb{E}_\tau$, consists of $N\times N$ complex random matrices with i.i.d.\ Gaussian entries of zero mean and correlations
\begin{equation}
\E[|X_{ij}|^2] = \frac{1}{N},
\qquad
\E[X_{ij}\overline{X_{ji}}] = \frac{\tau}{N},
\end{equation}
for $1\le i\neq j\le N$, with $\tau\in[-1,1]$ \cite{sommers1988spectrum,girko1986elliptic}.
\end{definition}

\begin{definition}[Haar unitary ensemble]\label{def:unitary}
The matrix $\vb{U}$ denotes an $N\times N$ random unitary matrix distributed according to the Haar measure on the unitary group $\mathrm{U}(N)$ \cite{mehta2004random}.
\end{definition}

\begin{definition}[Complex Wishart ensemble]\label{def:complex_wishart}
Let $\vb{A}$ be an $N\times T$ matrix with i.i.d.\ complex Gaussian entries of zero mean and unit variance. The matrix
\begin{equation}
\vb{W}_q = \frac{1}{T}\vb{A}\vb{A}^*
\end{equation}
is called a complex Wishart matrix with aspect ratio $q=N/T$ \cite{wishart1928generalised}.
\end{definition}

\begin{definition}[Two--ring invariant ensemble]\label{def:two_ring}
Let $\vb{D}$ be a $N\times N$ diagonal matrix whose entries are i.i.d. and uniformly distributed on the union of the two circles $\{|z|=r_1\}\cup\{|z|=r_2\}$ where $0\leq r_1<r_2$. We define
\begin{equation}
\vb{C}_{r_1,r_2} = \vb{U}\vb{D}\vb{U}^*,
\end{equation}
where $\vb{U}$ is Haar-distributed on $\mathrm{U}(N)$.
\end{definition}

We now compute the quantities required to apply the results proved in this paper. To this end, we adopt the notation and definitions introduced in~\cite{bousseyroux1}, in particular the transforms $\mathcal{H}$, $\tilde{\mathcal{R}}_1$, and $\tilde{\mathcal{R}}_2$, whose definitions we briefly recall below.

Let $\mathbf{M}$ be a random matrix of size $N\times N$, let $\mathbf{U}$ be an independent Haar-distributed unitary matrix, and let $\psi_1,\psi_2 \in \mathbb{C}^N$. We define the functional $\mathcal{H}_{\mathbf{M}}^N$ by
\begin{equation}\label{def:H}
\frac{1}{2N}\,
\log\!\Bigl\{
\E_{\mathbf{M},\mathbf{U}}\!\left[
\exp\!\Bigl(2N\,\Re\langle \psi_1,\, \mathbf{U}\mathbf{M}\mathbf{U}^* \psi_2\rangle\Bigr)
\right]
\Bigr\}
\;=\;
\mathcal{H}_{\mathbf{M}}^N\!\bigl(
\|\psi_1\|\,\|\psi_2\|,\,
\langle\psi_1,\psi_2\rangle,\,
\overline{\langle\psi_1,\psi_2\rangle}
\bigr),
\end{equation}\footnote{Here $\langle \cdot , \cdot \rangle$ denotes the standard inner product on $\C^N$,
$\langle x,y\rangle = \sum_{i=1}^N \overline{x_i}\,y_i$.}
We then set
\begin{equation}
\mathcal{H}_{\mathbf{M}} \;=\; \lim_{N\to\infty} \mathcal{H}_{\mathbf{M}}^N.
\end{equation}
Furthermore, we define the derivatives
\begin{equation}\label{R_transforms_def}
\mathcal{R}_1(\alpha,\beta) \;:=\; \partial_{\alpha}\,\mathcal{H}(\alpha,\beta),
\qquad
\mathcal{R}_2(\alpha,\beta) \;:=\; 2\partial_{\beta}\,\mathcal{H}(\alpha,\beta)\footnote{%
The derivative $\partial_{\beta}$ denotes the Wirtinger derivative, defined as
$\displaystyle \partial_{\beta} = \tfrac{1}{2}\bigl(\partial_{\Re(\beta)} - i\,\partial_{\Im(\beta)}\bigr)$.
}
\end{equation}
We also denote by $\widetilde{\mathcal{R}}_1$ and $\widetilde{\mathcal{R}}_2$ the multivalued analytic functions 
that collect all possible determinations (branches) of $\mathcal{R}_1$ and $\mathcal{R}_2$, respectively. We say that the functions $\mathcal{R}_1$ and $\mathcal{R}_2$ defined by Eqs.~\eqref{R_transforms_def} are the principal branches of $\widetilde{\mathcal{R}}_1$ and $\widetilde{\mathcal{R}}_2$.

Table~\ref{tab} summarizes these quantities for the random matrix ensembles described above. The elliptic Ginibre case has already been studied in~\cite{bousseyroux1}. The Ginibre ensemble corresponds simply to the elliptic case with $\tau=0$. For the Wishart ensemble $\vb{W}_q$, the computation of $\mathcal{H}_{\vb{W}_q}$ reduces to a Gaussian integral. In the unitary case, since the ensemble is bi-invariant, formula~\eqref{def:H} yields
\begin{equation}
    \mathcal{H}^{N}_{\vb{U}}(\alpha)
    =
    \frac{1}{2N}
    \log\!\Biggl(
    \E_{\vb{U}}
    \Bigl[
    \exp\!\bigl(2N\,\Re\!\bigl(\alpha [\vb{U}]_{11}\bigr)\bigr)
    \Bigr]
    \Biggr).
\end{equation}
Here, $[\vb{U}]_{11}$ denotes the first entry of the first column of $\vb{U}$. Since $[\vb{U}]_{11}$ follows a beta distribution, a saddle-point analysis leads to the expression reported in Table~\ref{tab}. The corresponding expressions for $\mathcal{R}_1$ and $\mathcal{R}_2$ are then obtained by differentiation. The presence of the square root $\sqrt{\cdot}$ implies the existence of two determinations, which is accounted for by the $\pm$ sign. The case of the two--ring ensemble $\vb{C}_{r_1,r_2}$ requires a completely different approach. Instead of computing $\mathcal{H}$ via the spherical integral, we rely on Result~4 of~\cite{bousseyroux1} to obtain the relevant determinations of $\tilde{\mathcal{R}}_1$ and $\tilde{\mathcal{R}}_2$ through the computation of the functions $\g(z)$ and $h(z)$, depending on whether $|z|>r_2$ or $r_1<|z|<r_2$. We emphasize that this method does not guarantee that all possible determinations are obtained. This example illustrates that knowing a single determination of $\partial_{\alpha}\tilde{\mathcal{R}}_1(0,\beta)$ or $\tilde{\mathcal{R}}_2(0,\beta)$ does not allow one to infer the others in a straightforward manner. Notice that the ensembles $\vb{G}$, $\vb{W}_q$, and $\vb{U}$ are bi-invariant, which implies $\tilde{\mathcal{R}}_2=0$.
\begin{table}[h!]
\centering
\scriptsize
\setlength{\tabcolsep}{3pt}
\renewcommand{\arraystretch}{2.5}

\begin{tabular}{|c|c|c|c|c|c|}
\hline
 & $\vb{G}$ & $\vb{E}_\tau$ & $\vb{U}$ & $\vb{W}_q$ & $\vb{C}_{r_1,r_2}$ \\
\hline

$\mathcal{H}$
& $\dfrac{\alpha^2}{2}$
& $\dfrac{\alpha^2}{2}+\dfrac{\tau}{4}(\beta^2+\overline{\beta}^2)$
& $\dfrac{-1+\sqrt{1+4\alpha^2}}{2}
+\dfrac12\ln\!\Bigl(\dfrac{-1+\sqrt{1+4\alpha^2}}{2\alpha^2}\Bigr)$
& $\displaystyle \frac{1}{2q}\ln\!\left(\frac{1}{|1-q\beta|^2-q^2\alpha^2}\right)$
& unknown
\\
\hline

$\tilde{\mathcal{R}}_1(\alpha, \beta)$
& $\alpha$
& $\alpha$
& $\displaystyle \frac{-1\pm\sqrt{1+4\alpha^2}}{2\alpha}$
& $\displaystyle \frac{q\alpha}{|1-q\beta|^2-q^2\alpha^2}$
& unknown
\\
\hline

$\tilde{\mathcal{R}}_2(\alpha, \beta)$
& $0$
& $\tau\,\beta$
& $0$
& $\displaystyle \frac{1-q\overline{\beta}}{|1-q\beta|^2-q^2\alpha^2}$
& unknown
\\
\hline

$\partial_{\alpha}\tilde{\mathcal{R}}_1(0,\beta)$
& $1$
& $1$
& $1$
& $\frac{q}{|1 - q\beta|^2}$
& \begin{tabular}[c]{@{}c@{}}
$\displaystyle \frac{r^2-r_1^2r_2^2|\beta|^2}{1-r^2|\beta|^2}, $ \\[0.5ex]
$\displaystyle
\frac{(4|\beta|^2r_1^2 - 1)(4|\beta|^2r_2^2 - 1)}{8 |\beta|^4 (r_1^2 - r_2^2)} - 1$
$\displaystyle, ...$ \\[0.5ex]
\end{tabular}
\\
\hline

$\tilde{\mathcal{R}}_2(0,\beta)$
& $0$
& $\tau\,\beta$
& $0$
& $\displaystyle \frac{1}{1-q\beta}$
& \begin{tabular}[c]{@{}c@{}}
$0, $ \\[0.5ex]
$\displaystyle -\frac{1}{2\beta}$
$\displaystyle, ...$ \\[0.5ex]
\end{tabular}
\\
\hline

\end{tabular}
\caption{Important quantities associated with the ensembles $\vb{G}$, $\vb{E}_\tau$, $\vb{U}$, $\vb{W}_q$, and $\vb{C}_{r_1,r_2}$. The parameter $r$ is defined by $r^2 = \frac{r_1^2 + r_2^2}{2}$ and we have assumed that $r_1<r_2$.}
\label{tab}
\end{table}

\section{Numerical examples}\label{examples}

\subsection{Lemniscate and bow tie}\label{sub:lemniscate}

\begin{proposition}\label{thm:lemniscate}
Let $\vb{B}$ be a bi-invariant random matrix whose inner and outer radii are denoted by $r_{-,\vb{B}}$ and $r_{+,\vb{B}}$, respectively, and let $\vb{D} = \diag(1, \ldots, 1, -1, \ldots, -1)$ be a diagonal matrix with half of the entries equal to $1$ and the other half equal to $-1$. Then, the outer boundary of the spectrum is given by the set of complex numbers $z \in \C$ such that
\begin{equation}\label{outer}
     |z^2 - 1|^2 = r_{+, \vb{B}}^2 (|z|^2 + 1).
\end{equation}
When this makes sense, the circle of radius $\sqrt{r_{-, \vb{B}}^2 - 1}$ is the inner spectral boundary.
\end{proposition}

\begin{remarks}
    \item The case $r_{+, \vb{B}} = 1$ yields the classical lemniscate curve given by
    \begin{equation}\label{eqlemniscate}
        |z^2 - 1|^2 = (|z|^2 + 1)^2.
    \end{equation} This result was originally obtained in \cite{BianeLehner2001}.
    \item Figure~\ref{fig:lemniscate} illustrates this result. The middle and right panels display the same outer boundary, highlighting the universal character of the spectral edges in the bi-invariant case, as emphasized in the second remark of Result~\ref{main_result}.
\end{remarks}

\begin{proof}
This follows directly from formula \eqref{bi-invariant-case}, since
\begin{equation}
    h_{\vb{D}}(z) = \frac{|z|^2 + 1}{|z^2 - 1|^2}.
\end{equation}
\end{proof}

\begin{figure}[h!]
    \centering
\includegraphics[width=\textwidth]{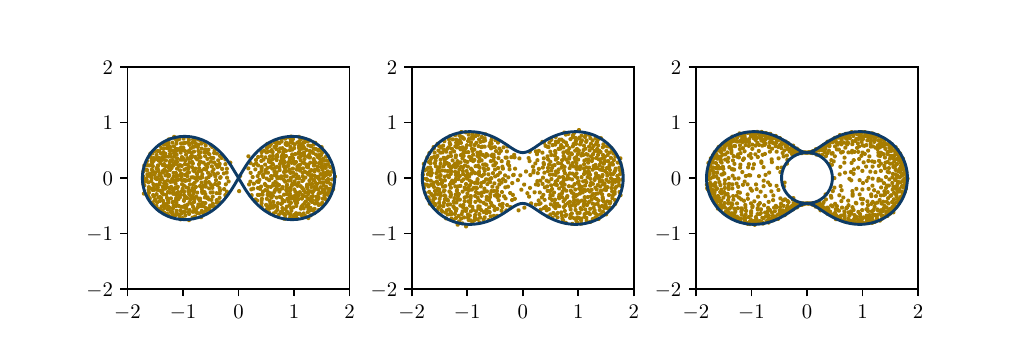}
    \caption{Empirical eigenvalues from single realizations of non-Hermitian deformations of $\vb{D}$.
Left: $\vb{D}+\vb{U}$ with $\vb{U}$ a random unitary matrix ($N=1000$); the lemniscate
\eqref{eqlemniscate} gives the theoretical spectral boundary.
Center: $\vb{D}+\sigma\vb{G}$ with $\sigma=1.1$ and $\vb{G}$ a complex Ginibre matrix
($N=1000$); the outer boundary is given by \eqref{outer} with $r_{+,\vb{B}}=\sigma$.
Right: $\vb{D}+\sigma\vb{U}$ with $\sigma=1.1$ and $\vb{U}$ a random unitary matrix
($N=1000$); besides the outer boundary \eqref{outer}, an inner spectral edge appears
at radius $\sqrt{\sigma^2-1}$.}
    \label{fig:lemniscate}
\end{figure}

\subsection{Airplane wing}\label{sub:airplane}

\begin{proposition}
Let $\vb{W}_q$ be a complex Wishart matrix defined in Definition \ref{def:complex_wishart} and let $\vb{B}$ be a large bi-invariant matrix such that $r_{-, \vb{B}} = 0$ and $r_{+, \vb{B}} = 1$. The spectral boundary of $\vb{M}$, defined by
\begin{equation}
    \vb{M} = \vb{W}_q + \sigma \vb{B},
\end{equation}
is given by the set of points $z \in \C$ satisfying
\begin{equation}\label{theo}
    \sigma^2 \Im\left(\frac{z + q - 1 \pm \sqrt{(z + q - 1)^2 - 4qz}}{2qz}\right) = - \Im(z).
\end{equation}
\end{proposition}

\begin{remarks}
    \item This can be viewed as an ``unflattened'' version of the Marchenko--Pastur distribution introduced in \cite{marchenko1967distribution}. Figure~\ref{fig:aile_avion} illustrates the case where $\vb{A}$ is a complex Ginibre matrix. For small values of $\sigma$, using equation~\eqref{traceMP}, the spectral boundaries are described by the union of the curves $x \mapsto \sigma^2 \pi \rho(x)$ and $x \mapsto -\sigma^2 \pi \rho(x)$, where $\rho$ denotes the Marchenko--Pastur density. One also observes that the smallest and largest real parts of the eigenvalues remain close to the standard Marchenko--Pastur edges
\begin{equation}
    \lambda_- = (1 - \sqrt{q})^2
    \qquad \text{and} \qquad
    \lambda_+ = (1 + \sqrt{q})^2.
\end{equation}
This behavior is precisely described by the last remark of Result~\ref{main_result}.
\end{remarks}

\begin{proof}
It suffices to use equation \eqref{formula_bi} together with the expression of $\g_{\vb{W}_q}$ computed in \cite{potters2020first}:
\begin{equation}
    \g_{\vb{W}_q}(z) = \frac{z + q - 1 \pm \sqrt{(z + q - 1)^2 - 4qz}}{2qz}.
\end{equation}
\end{proof}

\begin{figure}[h!]
    \centering
    \includegraphics{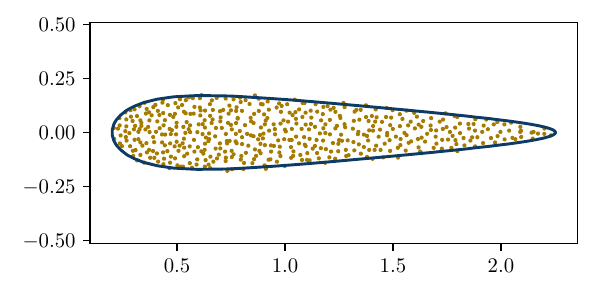}
    \caption{Complex eigenvalues of a matrix of size $N = 500$ given by $\vb{W}_q + \sigma \vb{G}$, where $\vb{W}_q$ is a complex Wishart matrix with parameter $q = 1/4$, $\sigma = 0.3$, and $\vb{G}$ is a complex Ginibre matrix. The theoretical spectral boundary is also shown using the equation \eqref{theo}.}
    \label{fig:aile_avion}
\end{figure}

\subsection{Samoussa}\label{sub:samoussa}

\begin{proposition}\label{propsamoussa}
    We consider two independent realizations of complex Wishart matrices $\vb{W}_1$ and $\vb{W}_2$ with parameter $q$ (see Definition~\ref{def:complex_wishart}), and we study the eigenvalues of the matrix
\begin{equation}
    \vb{M} = \vb{W}_1 + i \vb{W}_2.
\end{equation}

The spectral boundary is obtained as follows. We first define the curve $\mathcal{C}$ as the set of points $x$ satisfying
\begin{equation}
     \frac{q}{|x - q|^2} + \frac{q}{|x - iq|^2} = 1.
\end{equation}
We then consider the image of $\mathcal{C}$ under the map $\phi$ defined by
\begin{equation}
    \phi: x \mapsto x + \frac{x}{x - q} + \frac{ix}{x - iq}.
\end{equation}

\end{proposition}

\begin{remarks}
    \item Figure~\ref{fig:samoussa} shows that $\phi(\mathcal{C})$ accurately predicts the spectral boundary of the matrix $\vb{M}$.
\end{remarks}

\begin{proof}
Using the additivity of $\mathcal{R}_1$ and $\mathcal{R}_2$ (see \cite{bousseyroux1}), one can compute $\mathcal{R}_1$ and $\mathcal{R}_2$ of $\vb{M}$ using the expressions of the transforms of complex Wishart matrices given in Table~\ref{tab}, and then apply Result~\ref{thm:bordspectre}.
\end{proof}

\begin{figure}[h!]
    \centering
    \includegraphics[scale = 1]{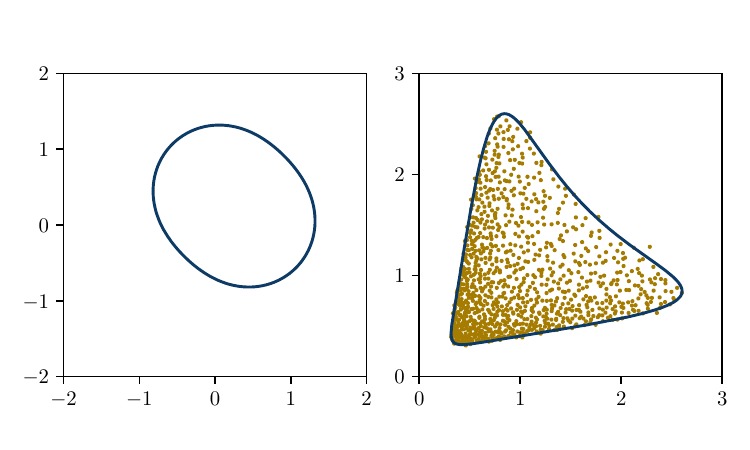}
    \caption{The right panel shows the eigenvalues from a single realization of the matrix $\vb{M}$ with $q = 0.5$ and $N = 1000$. We use the procedure described in Proposition~\ref{propsamoussa} to determine the spectral boundary. The left panel shows the curve $\mathcal{C}$ computed numerically, and $\phi(\mathcal{C})$ is drawn on the right as the predicted spectral boundary.}
    \label{fig:samoussa}
\end{figure}

\subsection{Eye of Sauron}\label{sub:sauron}

\begin{proposition}\label{soron}
Let $\vb{E}$ be a complex elliptic matrix with parameter $-1\leq \tau \leq 1$, and let $\vb{U}$ be unitary. For $\sigma \geq 0$, consider the matrix
\begin{equation}
    \vb{M} := \vb{U} + \sigma \vb{E}.
\end{equation}
The outer boundary of the limiting spectral distribution is an ellipse with semi-axes
\begin{equation}
    \frac{(1 \pm \tau)\sigma^2 + 1}{\sqrt{1 + \sigma^2}}.
\end{equation}
When $\sigma < 1$, the spectrum exhibits a hole in the form of a disk of radius $\sqrt{1 - \sigma^2}$.
\end{proposition}

\begin{remarks}
    \item When $\sigma = 0$, we recover the unit circle, which is the support of the limiting spectral distribution of a unitary matrix.
    \item Figure~\ref{fig:soron} shows a numerical example.
\end{remarks}

\begin{proof}
Using Table~\ref{tab}, we obtain
\begin{equation}
    \widetilde{\mathcal{R}_{1, \vb{M}}}(\alpha, \beta) = \frac{-1 \pm \sqrt{1 + 4\alpha^2}}{2\alpha} + \sigma^2 \alpha
\end{equation}
and
\begin{equation}
    \widetilde{\mathcal{R}_{2, \vb{M}}}(\alpha, \beta) = \tau \beta.
\end{equation}
There is therefore only one determination to consider for $\mathcal{R}_2$. By contrast, for $\mathcal{R}_1$, taking the branch with the plus sign yields
\begin{equation}
    \partial_{\alpha} \mathcal{R}_1(0, \beta) = 1 + \sigma^2,
\end{equation}
whereas taking the minus sign in front of the square root gives
\begin{equation}
    \mathcal{R}_1(\alpha, -z \alpha^2) \underset{\alpha \to 0}{=} \frac{-1}{\alpha} - (1 - \sigma^2)\alpha + o(\alpha).
\end{equation}

By Result~\ref{thm:bordspectre}, the outer spectral boundary of $\vb{M}$ is given by the image of the circle of radius $\sqrt{1 + \sigma^2}$ under the map $x \mapsto x + \tau/x$, which yields the desired ellipse. If there exists an inner radius, then by Result~\ref{thm:bordspectre}, a type~2 boundary must satisfy $|z|^2 = 1 - \sigma^2$. For small $\sigma$, such an inner boundary must exist since we are perturbing the unit circle spectrum of $\vb{U}$. Hence, for small $\sigma$, there exists a hole in the form of a disk of radius $\sqrt{1 - \sigma^2}$. By continuity, this formula holds up to $\sigma = 1$.
\end{proof}

\begin{figure}[h!]
    \centering
    \includegraphics[scale = 1]{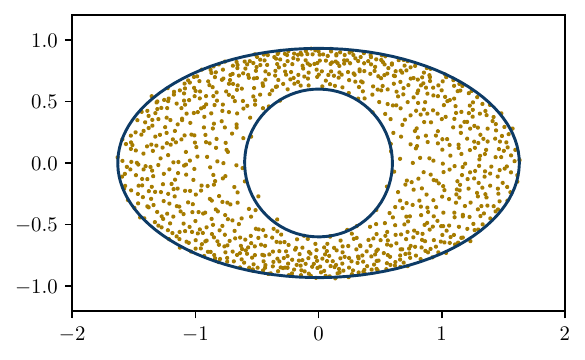}
    \caption{Empirical eigenvalues of $\vb{U} + \sigma \vb{E}$, where $\vb{U}$ is a random unitary matrix, $\sigma = 0.8$, and $\vb{E}$ is a complex elliptic matrix with parameter $\tau = 0.8$. The boundaries are given by Proposition~\ref{soron}.}
    \label{fig:soron}
\end{figure}
\subsection{Two rings triangularly deformed}\label{sub:rings}
We conclude with a more involved example. We consider the matrix $\vb{C}_{r_1,r_2}$ defined in Definition~\ref{def:two_ring}. We define
\begin{equation}\label{defM}
    \vb{M} := \vb{C}_{r_1,r_2} + \vb{D},
\end{equation}
where $\vb{D}$ is a diagonal matrix in which one third of the coefficients are equal to $1$, one third to $e^{2i\pi/3}$, and one third to $e^{4i\pi/3}$.

A priori, in order to apply Result~\ref{main_result} by viewing $\vb{M}$ as a deformation of $\vb{D}$ by $\vb{C}_{r_1,r_2}$, one would need to know all the possible determinations of the transforms associated with $\vb{C}_{r_1,r_2}$. Unfortunately, we were not able to compute all of them explicitly. Table~\ref{tab} nevertheless provides several such determinations. We conjecture that these are sufficient to recover all the spectral boundaries of $\vb{M}$.
\begin{conjecture}\label{conjecture}
Applying the \emph{type~1 boundary procedure} of Result~\ref{thm:bordspectre}
to the following pairs
$(\partial_{\alpha}\mathcal{R}_1(0,\beta),\,\mathcal{R}_2(0,\beta))$:

\begin{itemize}

\item
\begin{equation}\label{couple1}
\left(
\frac{r^2 - r_1^2 r_2^2 |\beta|^2}{1 - r^2 |\beta|^2},
\; 0
\right),
\end{equation}
where the parameter $r$ is defined by $r^2 = \frac{r_1^2 + r_2^2}{2}$,
yields the outer spectral boundary.

\item
\begin{equation}\label{couple2}
\left(
\frac{(4|\beta|^2 r_{\min}^2 - 1)(4|\beta|^2 r_{\max}^2 - 1)}
{8|\beta|^4 (r_{\max}^2 - r_{\min}^2)} - 1,
\; -\frac{1}{2\beta}
\right),
\end{equation}
where $r_{\max}=\max(r_1,r_2)$ and $r_{\min}=\min(r_1,r_2)$,
generates an additional spectral boundary, which may be empty if
Eq.~\eqref{sol1} admits no solution.

\item
When $\dfrac{1}{r_1^2}+\dfrac{1}{r_2^2}\leq 2$, there exists an inner spectral boundary
given by the circle centered at the origin with radius
\begin{equation}
\sqrt{\frac{r_1^2+r_2^2-1-\sqrt{(r_1^2-r_2^2)^2+1}}{2}}.
\end{equation}

\end{itemize}

The conjecture is that these constructions generate all spectral boundaries of $\vb{M}$.
\end{conjecture}

\begin{remarks}
\item Figure~\ref{fig:placeholder} illustrates the conjecture for several values of $(r_1,r_2)$.
\end{remarks}

\begin{proof}
This follows from an application of Result~\ref{main_result}, using the explicit
expressions of $\mathcal{R}_1$ and $\mathcal{R}_2$ associated with
$\vb{C}_{r_1,r_2}$ given in Table~\ref{tab}. We also use the formulas for
$\g_{\vb{D}}(z)$ and $h_{\vb{D}}(z)$:
\begin{equation}
\g_{\vb{D}}(z)=\frac{1}{3}\left(
\frac{1}{z-1}+\frac{1}{z-\alpha}+\frac{1}{z-\overline{\alpha}}
\right),
\end{equation}
and
\begin{equation}
h_{\vb{D}}(z)=\frac{1}{3}\left(
\frac{1}{|z-1|^2}+\frac{1}{|z-\alpha|^2}+\frac{1}{|z-\overline{\alpha}|^2}
\right),
\end{equation}
where $\alpha=e^{2i\pi/3}$.

To determine the possible type~2 spectral boundaries, Eq.~\eqref{edge2}
of Result~\ref{main_result} shows that they are given by the points
$z\in\mathbb{C}$ such that
\begin{equation}
\frac{1}{r_1^2-|z|^2}+\frac{1}{r_2^2-|z|^2}=2.
\end{equation}
Solving this equation yields
\begin{equation}
|z|=
\sqrt{\frac{r_1^2+r_2^2+1-\sqrt{(r_1^2-r_2^2)^2+1}}{2}},
\end{equation}
which is meaningful only when
\begin{equation}
\frac{1}{r_1^2}+\frac{1}{r_2^2}\leq 2.
\end{equation}
\end{proof}

\begin{figure}
    \centering
\includegraphics[width=\textwidth]{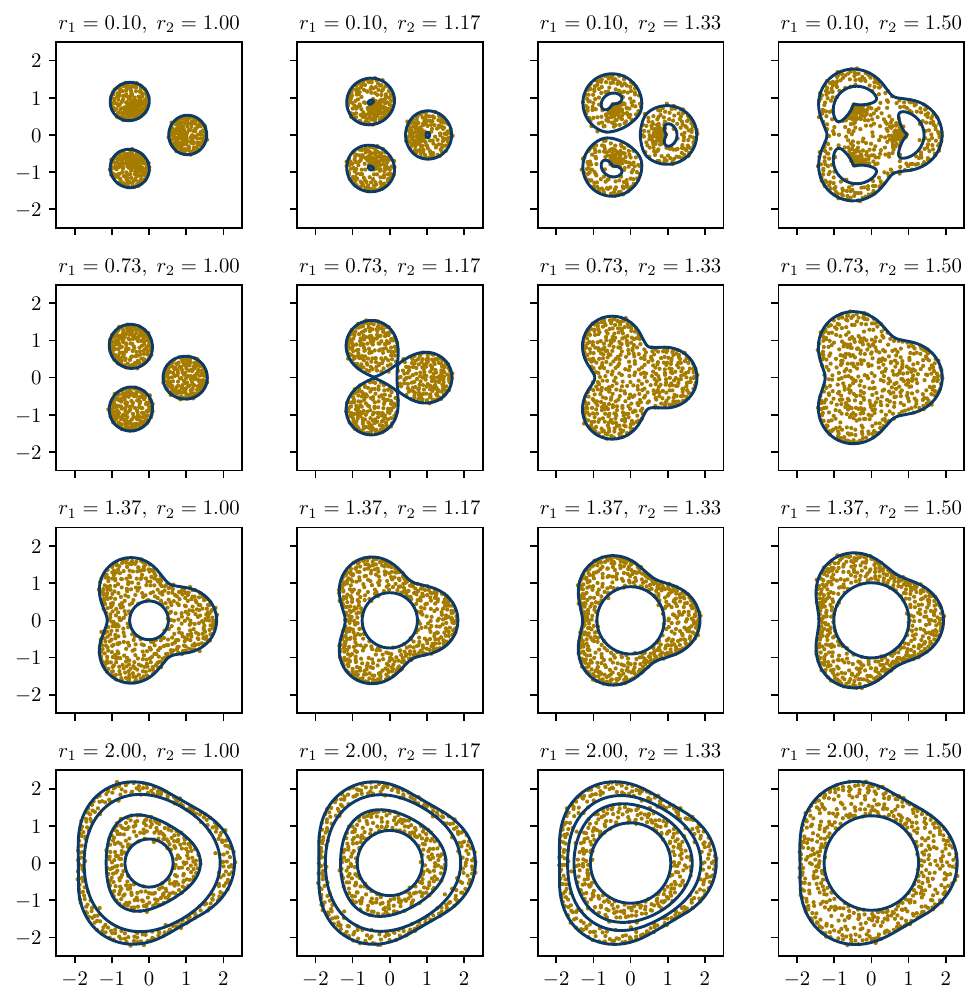}
    \caption{Complex eigenvalues of the matrix $\vb{M}$ defined in Eq.\eqref{defM} ($N=500$) for several parameter pairs $(r_1,r_2)$. The theoretical spectral boundaries predicted by Conjecture~\ref{conjecture} are also shown.}
    \label{fig:placeholder}
\end{figure}

\section*{Acknowledgements} 

We are grateful to Florent Benaych-Georges and Jean-Philippe Bouchaud for their valuable insights. This research was conducted within the Econophysics \& Complex Systems Research Chair, under the aegis of the Fondation du Risque, the Fondation de l’Ecole polytechnique, the Ecole polytechnique, and Capital Fund Management.

\bibliographystyle{plain}
\bibliography{References.bib}

@article{sarapin2025limiting,
  title={{Limiting Eigenvalue Distribution of the General Deformed Ginibre Ensemble}},
  author={Sarapin, Roman},
  journal={Journal of Statistical Physics},
  volume={192},
  number={8},
  pages={114},
  year={2025},
  publisher={Springer}
}

@article{fyodorov2025kac,
  title={{Kac-Rice inspired approach to non-Hermitian random matrices}},
  author={Fyodorov, Yan V},
  journal={arXiv preprint arXiv:2506.21058},
  year={2025}
}

@article{bousseyroux1,
  title={R-transforms for non-Hermitian matrices: a spherical integral approach},
  author={Bousseyroux, Pierre and Potters, Marc},
  journal={arXiv preprint arXiv:2601.09360},
  year={2026}
}

@article{bousseyroux2,
  title={The eigenvalues and eigenvectors of finite-rank normal perturbations of large rotationally invariant non-Hermitian matrices},
  author={Bousseyroux, Pierre and Potters, Marc},
  journal={arXiv preprint arXiv:2601.10427},
  year={2026}
}

@article{sommers1988spectrum,
  title={{Spectrum of large random asymmetric matrices}},
  author={Sommers, Hans Juergen and Crisanti, Andrea and Sompolinsky, Haim and Stein, Yaakov},
  journal={Physical review letters},
  volume={60},
  number={19},
  pages={1895},
  year={1988},
  publisher={APS}
}

@article{FeinbergZeeNPB501,
  author  = {Feinberg, J. and Zee, A.},
  title   = {{Non-Gaussian non-Hermitian random matrix theory: Phase transition and addition formalism}},
  journal = {Nucl. Phys. B},
  year    = {1997},
  volume  = {501},
  pages   = {643--669}
}

@article{FeinbergZeeNPB504,
  author  = {Feinberg, J. and Zee, A.},
  title   = {{Non-Hermitian random matrix theory: Method of hermitization}},
  journal = {Nucl. Phys. B},
  year    = {1997},
  volume  = {504},
  pages   = {579--608}
}

@article{ChalkerWangPRL1997,
  author  = {Chalker, J. T. and Wang, Z. Jane},
  title   = {{Diffusion in a random velocity field: spectral properties}},
  journal = {Phys. Rev. Lett.},
  year    = {1997},
  volume  = {79},
  pages   = {1797--1800}
}

@article{janik1997non,
  title={{Non-hermitian random matrix models}},
  author={Janik, Romuald A and Nowak, Maciej A and Papp, Gabor and Zahed, Ismail},
  journal={Nuclear Physics B},
  volume={501},
  number={3},
  pages={603--642},
  year={1997},
  publisher={Elsevier}
}

@article{JanikEtAlPRE1997,
  author  = {Janik, R. A. and Nowak, M. A. and Papp, G. and Wambach, J. and Zahed, I.},
  title   = {{Aspect of non-Hermitian random matrix models}},
  journal = {Phys. Rev. E},
  year    = {1997},
  volume  = {55},
  pages   = {4100--4107}
}

@article{Rodgers2010,
  author  = {Rodgers, G. J.},
  title   = {{Non-Hermitian random matrices and Brown measures}},
  journal = {J. Math. Phys.},
  year    = {2010},
  volume  = {51},
  pages   = {093304}
}

@article{liu2022phase,
  title={Phase transition of eigenvalues in deformed Ginibre ensembles},
  author={Liu, Dang-Zheng and Zhang, Lu},
  journal={arXiv preprint arXiv:2204.13171},
  year={2022}
}

@book{potters2020first,
  title={{A First Course in Random Matrix Theory: For Physicists, Engineers and Data Scientists}},
  author = {M.Potters, J.-P. Bouchaud},
  year={2020},
  publisher={Cambridge University Press}
}

@article{guionnet2011single,
  title={{The single ring theorem}},
  author={Guionnet, Alice and Krishnapur, Manjunath and Zeitouni, Ofer},
  journal={Annals of mathematics},
  pages={1189--1217},
  year={2011},
  publisher={JSTOR}
}

@misc{brown1986lidskii,
  title={{Lidskii’s theorem in the type II case, Geometric methods in operator algebras (Kyoto 1983), H. Araki and E. Effros (Eds.) Pitman Res. notes in Math. Ser 123}},
  author={Brown, LG},
  year={1986},
  publisher={Longman Sci. Tech}
}

@article{BianeLehner2001,
  author    = {Biane, Philippe and Lehner, Franz},
  title     = {{Computation of some examples of Brown's spectral measure in free probability}},
  journal   = {Colloquium Mathematicae},
  volume    = {90},
  number    = {2},
  pages     = {181--211},
  year      = {2001}
}

@article{haagerup2000brown,
  title={{Brown's spectral distribution measure for R-diagonal elements in finite von Neumann algebras}},
  author={Haagerup, Uffe and Larsen, Flemming},
  journal={Journal of Functional Analysis},
  volume={176},
  number={2},
  pages={331--367},
  year={2000},
  publisher={Elsevier}
}

@article{voiculescu1986addition,
  title={{Addition of certain non-commuting random variables}},
  author={Voiculescu, Dan},
  journal={Journal of functional analysis},
  volume={66},
  number={3},
  pages={323--346},
  year={1986},
  publisher={Elsevier}
}

@article{girko1985circular,
  title={{Circular law}},
  author={Girko, Vyacheslav L},
  journal={Theory of Probability \& Its Applications},
  volume={29},
  number={4},
  pages={694--706},
  year={1985},
  publisher={SIAM}
}

@article{girko1986elliptic,
  title={{Elliptic law}},
  author={Girko, VL},
  journal={Theory of Probability \& Its Applications},
  volume={30},
  number={4},
  pages={677--690},
  year={1986},
  publisher={SIAM}
}

@article{bercovici2022brown,
  title={{The Brown measure of a sum of two free random variables, one of which is R-diagonal}},
  author={Bercovici, Hari and Zhong, Ping},
  journal={arXiv preprint arXiv:2209.12379},
  year={2022}
}

@article{HallHo2022ImagSemicirc,
  author  = {Brian C. Hall and Ching-Wei Ho},
  title   = {{The Brown measure of the sum of a self-adjoint element and an imaginary multiple of a semicircular element}},
  journal = {Letters in Mathematical Physics},
  year    = {2022},
  volume  = {112},
  articleno = {19},
  note    = {Paper No.~19, 61 pp.},
  doi     = {10.1007/s11005-022-01516-3}
}

@article{Ho2022SelfAdjointPlusElliptic,
  author  = {Ching-Wei Ho},
  title   = {{The Brown measure of the sum of a self-adjoint element and an elliptic element}},
  journal = {Electronic Journal of Probability},
  year    = {2022},
  volume  = {27},
  pages   = {1--32},
  doi     = {10.1214/22-EJP840}
}

@article{HoYinZhong2025OutliersFullRankSingleRing,
  author       = {Ching-Wei Ho and Zhi Yin and Ping Zhong},
  title        = {{Outlier eigenvalues for full rank deformed single ring random matrices}},
  journal      = {arXiv preprint},
  eprint       = {2502.10796},
  archivePrefix= {arXiv},
  primaryClass = {math.PR},
  doi          = {10.48550/arXiv.2502.10796},
  year         = {2025}
}

@article{HoZhong2025DeformedSingleRing,
  author  = {Ching-Wei Ho and Ping Zhong},
  title   = {{Deformed single ring theorems}},
  journal = {Journal of Functional Analysis},
  year    = {2025},
  volume  = {288},
  number  = {5},
  pages   = {110797},
  doi     = {10.1016/j.jfa.2024.110797},
  note    = {Issue date: 1 Mar 2025}
}

@article{Zhong2021EllipticPlusFree,
  author       = {Ping Zhong},
  title        = {{Brown measure of the sum of an elliptic operator and a free random variable in a finite von Neumann algebra}},
  journal      = {arXiv preprint},
  eprint       = {2108.09844},
  archivePrefix= {arXiv},
  primaryClass = {math.OA},
  doi          = {10.48550/arXiv.2108.09844},
  year         = {2021},
  note         = {v5 (Aug 2025); to appear in American Journal of Mathematics}
}

@article{marchenko1967distribution,
  title={{Distribution of eigenvalues for some sets of random matrices}},
  author={Marchenko, Vladimir Alexandrovich and Pastur, Leonid Andreevich},
  journal={Matematicheskii Sbornik},
  volume={114},
  number={4},
  pages={507--536},
  year={1967},
  publisher={Russian Academy of Sciences, Steklov Mathematical Institute of Russian~…}
}

@article{haagerup2007brown,
  title={{Brown measures of unbounded operators affiliated with a finite von Neumann algebra}},
  author={Haagerup, Uffe and Schultz, Hanne},
  journal={Mathematica Scandinavica},
  pages={209--263},
  year={2007},
  publisher={JSTOR}
}

@article{wishart1928generalised,
  title={The generalised product moment distribution in samples from a normal multivariate population},
  author={Wishart, John},
  journal={Biometrika},
  volume={20},
  number={1/2},
  pages={32--52},
  year={1928},
  publisher={JSTOR}
}

@book{mehta2004random,
  title={Random matrices},
  author={Mehta, Madan Lal},
  volume={142},
  year={2004},
  publisher={Elsevier}
}

\appendix
\newpage
\section{APPENDICES}

\subsection{Proof of the first result} \label{first_result}

Once again, we strongly encourage the reader to consult~\cite{bousseyroux1}, where the notions we will use here are developed in detail.  
Let $\vb{M}$ be a large random matrix. We will use the following notations:  

\begin{equation}
    \g_{1, \vb{M}}(\omega, z) = \tau\!\Bigl[\omega\bigl(\omega^2 \vb{1}-(z\vb{1}-\vb{M})(z\vb{1}-\vb{M})^*\bigr)^{-1}\Bigr],
\end{equation}
and
\begin{equation}
    \g_{2, \vb{M}}(\omega, z) = -\tau\!\Bigl[(z\vb{1}-\vb{M})^*\bigl(\omega^2\vb{1}-(z\vb{1}-\vb{M})(z\vb{1}-\vb{M})^*\bigr)^{-1}\Bigr],
\end{equation}
as well as
\begin{equation}
    \mathcal{G}_{\vb{M}}(\omega,z) \;=\; 
    \begin{pmatrix}
        \g_{1, \vb{M}}(\omega, z) & \g_{2, \vb{M}}(\omega, z) \\
        \overline{\g_{2, \vb{M}}(\omega, z)} &\g_{1, \vb{M}}(\omega, z)
    \end{pmatrix}.
\end{equation}
As recalled in~\cite{bousseyroux1}, by taking $\omega = -i\epsilon$ with $\epsilon>0$ and $\epsilon \to 0$, we have $\g_{1, \vb{M}}(\omega, z) \to i\oo_{\vb{M}}(z)$, where $\oo_{\vb{M}}(z)$, defined in~\cite{bousseyroux1}, measures the non-normality of the eigenvectors. We will not use this quantity explicitly, except for the fact that $\oo_{\vb{M}}(z) \neq 0$ if and only if $\rho_{\vb{M}}(z) \neq 0$, where $\rho_{\vb{M}}(z)$ is the limiting spectral distribution of $\vb{M}$ evaluated at the point $z$. The idea is to look at the main result of~\cite{bousseyroux1}, which we restate here:

\begin{result}\label{main_result}
Let $\vb{A}$ be a large deterministic matrix and $\vb{B}$ a rotationally invariant random matrix. Then, for large $\omega$ and $z\in\C$, one expects
\begin{equation}\label{eq_result}
    \mathcal{G}_{\vb{A} + \vb{B}}(\omega, z)
    \;=\; \mathcal{G}_{\vb{A}}\!\left(\omega - \mathcal{R}_{1, \vb{B}}(\g_1,\g_2),\, z - \mathcal{R}_{2, \vb{B}}(\g_1,\g_2)\right)
\end{equation}where $\g_1 = \g_{1, \vb{A} + \vb{B}}$ and $\g_2 = \g_{2, \vb{A} + \vb{B}}$. 
This identity can be analytically continued to all $(\omega,z)$, although one must carefully choose the appropriate branches of $\mathcal{R}_1$ and $\mathcal{R}_2$.
\end{result}
We now aim to prove Result~\ref{main_result} by relying on the result.  
We denote $\vb{M} = \vb{A} + \vb{B}$.  
Let $z \in \supp(\rho_{\vb{M}})$, and consider the limit $z \to z'$, where $z'$ lies on the boundary of the spectral distribution of $\vb{M}$.  In this case, we have $\oo_{\vb{M}}(z) \underset{z \to z'}{\longrightarrow} 0$.  Taking the limit $\omega \to 0$ in~\eqref{eq_result} yields:
\begin{equation}\label{1step}
    \mathcal{G}_{\vb{M}}(0, z)
    \;=\;
    \mathcal{G}_{\vb{A}}\!\left(
        -\,\mathcal{R}_{1, \vb{B}}\big(i\oo_{\vb{M}}(z), \g_{\vb{M}}(z)\big),
        \; z - \mathcal{R}_{2, \vb{B}}\big(i\oo_{\vb{M}}(z), \g_{\vb{M}}(z)\big)
    \right).
\end{equation}
We now consider two cases.

\paragraph{First case.}  
We assume here that $\partial_{\alpha} \mathcal{R}_{1, \vb{B}}(0,\g_{\vb{M}}(z'))$ exists. 
 Then using \eqref{1step}, we get
\begin{equation}\label{eq45}
    \mathcal{G}_{\vb{M}}(0, z) 
    \;\underset{z\to z'}{=}\; \mathcal{G}_{\vb{A}}\!\left( - \partial_{\alpha} \mathcal{R}_{1, \vb{B}}(0,\g_{\vb{M}}(z'))\, i\oo_{\vb{M}}(z),\, z' - \mathcal{R}_{2, \vb{B}}(0,\g_{\vb{M}}(z'))\right).
\end{equation}
Looking only at the upper-left coefficient, we obtain:
\begin{equation}
    i\oo(z)
    \;\underset{z\to z'}{=}\; \g_1\!\left( - \partial_{\alpha} \mathcal{R}_{1, \vb{B}}(0,\g_{\vb{M}}(z'))\, i\oo_{\vb{M}}(z),\, z - \mathcal{R}_{2, \vb{B}}(0,\g_{\vb{M}}(z))\right),
\end{equation}
and therefore necessarily:
\begin{equation}
    i\oo_{\vb{M}}(z) \underset{z\to z'}{\sim} h_{\vb{A}}(z - \mathcal{R}_{2, \vb{B}}(0,\g_{\vb{M}}(z)))\, \partial_{\alpha} \mathcal{R}_{1, \vb{B}}(0,\g_{\vb{M}}(z'))\, i\oo_{\vb{M}}(z).
\end{equation}
Finally,
\begin{equation}\label{bord1}
    1 = h_{\vb{A}}(z' - \mathcal{R}_{2, \vb{B}}(0,\g_{\vb{M}}(z')))\, \partial_{\alpha} \mathcal{R}_{1, \vb{B}}(0,\g_{\vb{M}}(z')).
\end{equation}
Now, looking at the upper-right coefficient of equation~\eqref{eq45}, we obtain
\begin{equation}
    \g_{\vb{M}}(z') = \g_{\vb{A}}(z' - \mathcal{R}_{2, \vb{B}}(0,\g_{\vb{M}}(z'))).
\end{equation}
Setting $x = z' - \mathcal{R}_{2, \vb{B}}(0,\g_{\vb{M}}(z'))$ and substituting into equation~\eqref{bord1}, we obtain
\begin{equation}
    \partial_{\alpha} \mathcal{R}_{1, \vb{B}}(0,\g_{\vb{A}}(x)) = \frac{1}{h_{\vb{A}}(x)}.
\end{equation}
To return to the variable $z'$, we simply note that
\begin{equation}
    z' = x + \mathcal{R}_{2, \vb{B}}(0, \g_{\vb{M}}(z')) = x + \mathcal{R}_{2, \vb{B}}(0, \g_{\vb{A}}(x)),
\end{equation}
and we thus recover the first part of the result.

\paragraph{Second case.}

Consider now the case where $|\mathcal{R}_{1, \vb{B}}(0,\g_{\vb{M}}(z'))| \to +\infty$. From equation~\eqref{1step}, this forces us to expand $\g_{1, \vb{A}}(\omega, z)$ as $\omega \to +\infty$ and to extract the next term of order $\frac{1}{\omega^3}$.  
One then realizes that what ultimately matters are only the quantities $\tau(\vb{A})$ and $\tau(\vb{A}\vb{A}^*)$.  
Thus, to determine the spectral boundary, it suffices to replace $\vb{A}$ by, for instance, a matrix of the form $\tau(\vb{A})\vb{1} + \sqrt{\tau(\vb{A}\vb{A}^*) - |\tau(\vb{A})|^2}\vb{G}$, where $\vb{G}$ is a Ginibre matrix defined in Definition~\ref{def:ginibre}. This matrix is bi-invariant and satisfies $r_{+, \vb{G}} = 1$ and $r_{-, \vb{G}} = 0$ according to the circular law~\cite{girko1985circular}. In this case, the result of Result~1 applies directly — more precisely, through the translation by $\tau(\vb{A})$ from Eq.~\eqref{bi-invariant-case}, using the $\mathcal{R}_1$ and $\mathcal{R}_2$ transforms of $\vb{G}$ computed in Table~\ref{tab}.

\subsection{Proof of the second result}\label{second_result}
Let $\vb{B}$ be a large random matrix, a priori non-Hermitian.  
We use the following equations from~\cite{bousseyroux1}:
\begin{equation}\label{main1}
    i\,\oo(z)
    \;=\;
    \frac{-\,\mathcal{R}_{1, \vb{B}}(i\,\oo(z), \g(z))}
    {\mathcal{R}_{1, \vb{B}}(i\,\oo(z), \g(z))^2
    - \left|z - \mathcal{R}_{2, \vb{B}}(i\,\oo(z), \g(z))\right|^2},
\end{equation}
\begin{equation}\label{main2}
    \overline{\g(z)}
    \;=\;
    \frac{\mathcal{R}_{2, \vb{B}}(i\,\oo(z), \g(z)) - z}
    {\mathcal{R}_{1, \vb{B}}(i\,\oo(z), \g(z))^2
    - \left|z - \mathcal{R}_{2, \vb{B}}(i\,\oo(z), \g(z))\right|^2}.
\end{equation}
 Here $z \in \C$, $\mathcal{R}_{1,\vb{B}} \in \widetilde{\mathcal{R}}_{1,\vb{B}}$, $\mathcal{R}_{2,\vb{B}} \in \widetilde{\mathcal{R}}_{2,\vb{B}}$, and $\g(z) = \tau\!\bigl[(z\vb{1} - \vb{B})^{-1}\bigr]$. The functions $\oo(z)$ satisfy the property  $\oo(z) \neq 0 \;\Longleftrightarrow\; \rho(z) \neq 0$ where $\rho$ denotes the limiting spectral distribution of $\vb{B}$.

Let $z'$ be a point on the spectral boundary of $\vb{B}$. The idea is again to study the limit of the previous equations as $z \to z'$, knowing that $\oo(z) \to 0$ in this limit. The left-hand side of~\eqref{main1} tends to zero, so there are two possible scenarios:

\begin{equation}
    \mathcal{R}_{1, \vb{B}}(i\,\oo(z), \g(z)) \;\to\; 0,
\end{equation}
or
\begin{equation}
    \bigl|
    \mathcal{R}_{1, \vb{B}}(i\,\oo(z), \g(z))^2
    - \bigl|z - \mathcal{R}_{2, \vb{B}}(i\,\oo(z), \g(z))\bigr|^2
    \bigr|
    \;\to\; +\infty.
\end{equation}

\paragraph{First case.} We first consider the former case. We then have
\begin{equation}
    i\,\oo(z)
    \underset{z\to z'}{\sim}
    \frac{\mathcal{R}_{1, \vb{B}}(i\,\oo(z), \g(z))}
    {\bigl|z - \mathcal{R}_{2, \vb{B}}(i\,\oo(z), \g(z))\bigr|^2},
\end{equation}
which implies that
\begin{equation}
    \bigl|z' - \mathcal{R}_{2, \vb{B}}(0, \g(z'))\bigr|^2
    = \partial_{\alpha} \mathcal{R}_{1, \vb{B}}(0, \g(z')).
\end{equation}
The second equation then gives
\begin{equation}
    \g(z') = \frac{1}{z' - \mathcal{R}_2(0, \g(z'))}.
\end{equation}
This yields the first part of the result.

\bigskip

\paragraph{Second case.} Let's assume that $\tau(\vb{B}) = 0$. Let us now turn to the second case, where
\begin{equation}
    \bigl|
    \mathcal{R}_{1, \vb{B}}(i\,\oo(z), \g(z))^2
    - \bigl|z - \mathcal{R}_{2, \vb{B}}(i\,\oo(z), \g(z))\bigr|^2
    \bigr|
    \;\to\; +\infty.
\end{equation}
In this case, $|\mathcal{R}_{1, \vb{B}}(i\,\oo(z), \g(z))| \to +\infty$.  
Using~\eqref{main1}, we deduce that
\begin{equation}\label{eq3}
    \mathcal{R}_{1, \vb{B}}(i\,\oo(z), \g(z))
    \underset{z\to z'}{\approx}
    \frac{-1}{i\,\oo(z)}
    - \bigl|z - \mathcal{R}_{2, \vb{B}}(i\,\oo(z), \g(z))\bigr|^2 i\,\oo(z).
\end{equation}
The second equation~\eqref{main2} in turn gives
\begin{equation}
    \g(z) \underset{z\to z'}{\sim} -\,z'\,\bigl(i\,\oo(z)\bigr)^2.
\end{equation}
Substituting this into~\eqref{eq3}, we find
\begin{equation}
    \mathcal{R}_{1, \vb{B}}\!\bigl(i\,\oo(z), -z'\,(i\,\oo(z))^2\bigr)
    \underset{z\to z'}{\approx}
    \frac{-1}{i\,\oo(z)} - |z'|^2 i\,\oo(z),
\end{equation}
which establishes the second part of the result.

\end{document}